\documentclass[11pt]{article}

\usepackage{amsmath}
\usepackage{amsfonts}
\usepackage{amsthm}
\usepackage{amssymb}
\usepackage[utf8]{inputenc}

\usepackage{textcomp}

\usepackage[margin=1in]{geometry}
\usepackage{indentfirst}

\usepackage{abstract}

\usepackage{pbox}
\usepackage{comment}

\usepackage[linesnumbered, boxruled]{algorithm2e}
\usepackage{pgfplots}
\pgfplotsset{compat=1.4}
\usetikzlibrary{shapes}
\usetikzlibrary{plotmarks}

\usepackage{framed}
\usepackage{mdframed}
\usepackage{placeins}
\usepackage{subcaption}
\usepackage{afterpage}
\usepackage{hyperref}
\usepackage{courier}
\usepackage{multirow}
\usepackage{url}
\usepackage{cite}

\usetikzlibrary{patterns}

\allowdisplaybreaks

\newcommand{\eps}{\varepsilon}

\newcommand{\rr}{\mathbb{R}}

\usepackage{soul}

\usepackage[colorinlistoftodos,textsize=small,color=red!25!white,obeyFinal]{todonotes}

\newtheorem{theorem}{Theorem}

\newtheorem{corollary}{Corollary}
\newtheorem{lemma}{Lemma}
\newtheorem{definition}{Definition}

\usepackage{amsmath}

\newtheorem*{definition*}{Definition}
\newtheorem*{theorem*}{Theorem}

\newtheorem*{lemma*}{Lemma}

\title{Opponent Indifference in Rating Systems:\\ A Theoretical Case for Sonas}
\author{Greg Bodwin and Forest Zhang \\ \texttt{\{bodwin, forestz\}@umich.edu}\\
University of Michigan EECS}

\date{}

\begin{document}
\maketitle

\begin{abstract}
In competitive games, it is common to assign each player a real number \emph{rating} signifying their skill level.
A \emph{rating system} is a procedure by which player ratings are adjusted upwards each time they win, or downwards each time they lose.

Many matchmaking systems give players some control over their opponent's rating; for example, a player might be able to selectively initiate games against opponents whose ratings are publicly visible, or abort a game without penalty before it begins but after glimpsing their opponent's rating.
It is natural to ask whether one can design a rating system that does not incentivize a rating-maximizing player to act strategically, seeking games against opponents of one rating over another.
We show the following:
\begin{itemize}
    \item The \emph{full} version of this ``opponent indifference'' property is unfortunately too strong to be feasible.
    Although it is satisfied by some rating systems, these systems lack certain desirable expressiveness properties, suggesting that they are not suitable to capture most games of interest.
    
    \item However, there is a natural relaxation, roughly requiring indifference between any two opponents who are both ``reasonably evenly matched'' with the choosing player.
    We prove that this relaxed variant of opponent indifference, which we call $P$ opponent indifference, is viable.
    In fact, a certain strong version of $P$ opponent indifference precisely characterizes the rating system \emph{Sonas}, which was originally proposed for its empirical predictive accuracy on the outcomes of high-level chess games.
\end{itemize}
\end{abstract}

\setcounter{page}{0}

\thispagestyle{empty}

\clearpage

\section{Introduction}





For two-player competitive games like chess, professional sports, online gaming, etc, it is common practice to assign players/teams a real-number \emph{rating} capturing their skill level.
Ratings are commonly applied in matchmaking systems \cite{AL09, AM17}; some more creative recent applications include pricing algorithms \cite{YDM14}, item response theory \cite{Forisek09}, etc.\ (see Section \ref{sec:priorwork} for more detail on the contribution of these papers).
A \emph{rating system} is an algorithm that adjusts a player's rating upwards after each win, or downwards after each loss.
Some notable rating systems used in practice include Harkness \cite{Harkness67}, Elo \cite{Elo67}, Glicko \cite{Glickman95}, Sonas \cite{Sonas02}, TrueSkill \cite{HMG06}, URS \cite{UniversalRating}, and more.
We refer to the papers \cite{Glickman95, MCZ18, HMG06}, which each develop a notable rating system, for high-quality technical overviews of the inner workings and contrasts between the rating systems in the literature.

\subsection{Our Model of Rating Systems}

The exact model of rating systems studied in this paper is original, although really all of its ingredients are slight variants, rephrasings, or abstractions of ingredients that appear in prior work (we will address these similarities as we go).
Our goal is to consider a maximally general model of rating systems that are:
\begin{itemize}
\item \emph{One-dimensional}, meaning that they only track a rating parameter for each player and no additional statistics of performance,
\item \emph{Memoryless}, meaning that when the system updates players' ratings it considers only the game that was just played and not the history of previous game outcomes, and
\item \emph{Zero-sum}, meaning that the number of rating points gained by the winner of a game is always equal to the number of rating points lost by the loser.
\end{itemize}

To be clear, many rating systems used in theory or practice do \emph{not} satisfy these properties for various reasons.
Typically, this is because they intentionally augment a ``base'' rating system (which often does satisfy these three properties) to consider additional player statistics, game history, etc., in a way that trades simplicity and interpretability for improved empirical predictive accuracy of game outcomes \cite{Glickman95, CJ16, MRM07, CSEN17, DHMG07, NS10, MCZ18}.
Thus, this paper aims to study the space of ``base'' rating systems on which these more complex but practically accurate rating systems are built.

Additionally, a few of our results will consider the following property of rating systems:
\begin{itemize}
\item \emph{Translation-invariant}, meaning that the win probabilities for two players are determined entirely by the additive difference between their ratings.  (See Definition \ref{def:translationinvariant} for a formal definition.)
\end{itemize}
Translation invariance is common in practice; it implies that rating differences carry the same meaning for low-rated players as for high-rated players, and thus it makes the system more interpretable.
Many systems specifically follow the convention that a $200$-point rating difference maps to $0.75$ win probability for the higher-rated player; this rule was first applied by the rating system ELO \cite{Elo67}.

The goal of this paper is instead to continue a recent line of work on \emph{incentive-compatibility} as a property of rating systems.
The idea is that players experience rating as an incentive, and they will sometimes take strategic action to maximize their rating in the system.
For example, Glicko-2 was recently subject to an attack called \emph{volatility hacking} in which players achieved long-run boosts in their rating by deliberately losing certain games \cite{EL21}.
Some modern rating systems like URS address the possibility of strategic action by leaving parameters of their rating formula unpublished.\footnote{Jeff Sonas, personal communication, 2022}
This has several disadvantages; for example, players are less able to predict their rating changes, which in turn makes it harder to detect transcription errors in game outcome reporting. 
Similarly, bugs in the implementation of a rating system may be harder to detect.
It is desirable to achieve notions of incentive compatibility, even when the rating system in play is entirely public knowledge.

Our focus is on a new rating system property, informally stating that the system does not incentivize a rating-maximizing player to selectively seek games against opponents of one rating over another (even when the player currently feels over- or underrated by the system).
One of our main conceptual contributions is just to \emph{define} this property in a reasonable way: as we discuss shortly, the natural first attempt at a definition turns out to be too strong to be feasible.
We then describe a different version of the definition, and in support of this definition, we show that there is a natural class of rating systems (most notably including \emph{Sonas}, described shortly) that satisfy this alternate definition.
In order to explain this progression in ideas formally, we next give details on our model of rating systems.

\subsubsection{Formal Model of Rating Systems}
We follow the common basic model that each player has a hidden ``true rating,'' as well as a visible ``current rating'' assigned by the system.
Although current ratings naturally fluctuate over time, the goal of a rating system is to assign current ratings that are usually a reasonable approximation of true ratings.
When two players play a game, their respective true ratings determine the probability that one beats the other, and their respective current ratings are used by the system as inputs to a function that determines the changes in the ratings of the winner and the loser.

Concretely, a rating system is composed of two ingredients.
The first is the \emph{skill curve} $\sigma$, which encodes the system's model of the probability that a player of true rating $x$ will beat a player of true rating $y$.
The second is the \emph{adjustment function} $\alpha$, which takes the players' current ratings as inputs, and encodes the number of points that the system awards to the first player and takes from the second player in the event that the first player beats the second player.
We are not allowed to pair together an \emph{arbitrary} skill curve and adjustment function; rather, a basic fairness axiom must be satisfied.
This axiom is phrased in terms of an \emph{expected gain function}, which computes the expected number of rating points gained by a player of current rating $x$ and true rating $x^*$, when they play a game against a player of current rating $y$ and true rating $y^*$.
The fairness axiom states that, when any two correctly-rated players play a game, their expected rating change should be $0$.

\begin{definition} [Rating Systems] \label{def:ratingsystems}
A \textbf{rating system} is a pair $(\sigma, \alpha)$, where:
\begin{itemize}
\item $\sigma : \rr^2 \to [0, 1]$, the \textbf{skill curve}, is a continuous function that is weakly increasing in its first parameter and weakly decreasing in its second parameter.
We assume the game has no draws, and thus $\sigma$ must satisfy\footnote{We ignore the possibility of draws in this exposition purely for simplicity.  In the case of draws, one could interpret a win as $1$ victory point, a draw as $1/2$ a victory point, and a loss as $0$ victory points, and then interpret $\sigma(x, y)$ as the expected number of victory points scored by a player of true rating $x$ against a player of true rating $y$. See \cite{Kovalchik20} for more discussion, and for a rating systems that can handle even finer-grained outcomes than the win/loss/draw trichotomy.}
\begin{align*}
\sigma(x, y) + \sigma(y, x) = 1 \tag*{for all $x, y$. ($\sigma$ Draw-Free)}
\end{align*}
\item $\alpha : \rr^2 \to \rr_{\ge 0}$ is the \textbf{adjustment function}. We assume that every game has the possibility of changing the rating of one of the players, and thus $\alpha$ must satisfy
\begin{align*}
\alpha(x, y) + \alpha(y, x) > 0 \tag*{for all $x, y$. (Consequential)}
\end{align*}
\item Let $\gamma : \rr^4 \to \rr$ be the \textbf{expected gain function}, defined by
$$\gamma(x,x^* \mid y,y^*) := \alpha(x,y)\sigma(x^*,y^*)-\alpha(y,x)\sigma(y^*,x^*).$$
Then we require the following fairness axiom:
$$\gamma(x, x \mid y, y) = 0 \text{ for all } x, y.$$
\end{itemize}
\end{definition}

The fairness axiom is only implicit in prior work.
Other rating systems, including ELO, Glicko, etc., \cite{Glickman95, MCZ18, HMG06}, often define their rating system by giving a skill curve and a \emph{$K$-function}, which together imply an adjustment function in a sense we next explain.
An adjustment function can be implicitly defined in this way if and only if it satisfies this fairness axiom.
It will be technically convenient for us to be able to have both the explicit adjustment function $\alpha$ and the $K$-function to reference in our proofs.

\subsubsection{$K$-Functions}
Rearranging the fairness axiom, we get the identity
\begin{align*}
    \sigma(x, y) &= \frac{\alpha(y, x)}{\alpha(x, y) + \alpha(y, x)}.
\end{align*}
Thus the adjustment function determines the skill curve (in the sense that, for a given adjustment function $\alpha$, there is a unique skill curve $\sigma$ that satisfies the fairness axiom).
However, when designing a rating system, it is more intuitive to model the game with a skill curve first and then pick an adjustment function second.
The skill curve does \emph{not} fully determine the adjustment function; rather, the value of $\sigma(x, y)$ only determines the \emph{ratio} between $\alpha(x, y)$ and $\alpha(y, x)$.
The \emph{$K$-function} is an auxiliary function that determines the scaling.
We define it as the denominator part of the previous identity.\footnote{One can also interpret $K$-functions through a gambling model where the players respectively put $\alpha(x, y), \alpha(y, x)$ of their rating points into a pot, and then the winner takes all the rating points in the pot.  The $K$-function $K(x, y) = K(y, x)$ is the pot size.}
\begin{definition} [$K$-Functions]
For a rating system $(\sigma, \alpha)$, its associated \textbf{$K$-Function} is defined by $K(x, y) := \alpha(x, y) + \alpha(y, x)$.
\end{definition}

Note that being consequential means that $K(x,y)>0$ for all $x,y$.

Many rating systems in the literature do not \emph{explicitly} define an adjustment function; rather, they define their adjustment function \emph{implicitly} by instead giving the skill curve and the $K$-function \cite{Glickman95, Elo67}.
Most notably, Elo uses \emph{$K$-factors}, which are a special case of our $K$-functions.\footnote{We have chosen the name \emph{$K$-function} to indicate that they generalize $K$-factors.  $K$ does not stand for anything; is an arbitrary variable name in the Elo adjustment function formula.}
See \cite{Sonas11} for a discussion of the practical considerations behind choosing a good $K$-function.

\subsection{The Elo and Sonas Rating Systems}

The most famous and popular rating system used in practice is probably Elo.
Its basis is a logistic skill curve: in a game between two players of true ratings $x, y$, Elo maps the quantity $x-y$ to a win probability for the first player using a logistic function.
Many implementations of ELO, like the one used by the chess federation FIDE, also include a \emph{thresholding rule}: e.g., if $x-y > 400$, then the FIDE implementation instead treats $x-y$ as exactly $400$.
In other words, it assumes the game is chaotic enough that no player ever has more than a 96\% win probability over another.
This thresholding also regularizes against extreme rating swings.

\begin{figure}[h]
    \centering
    \includegraphics[scale=0.4]{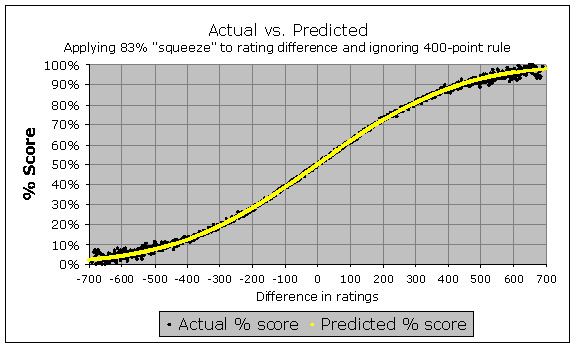}
    \includegraphics[scale=0.4]{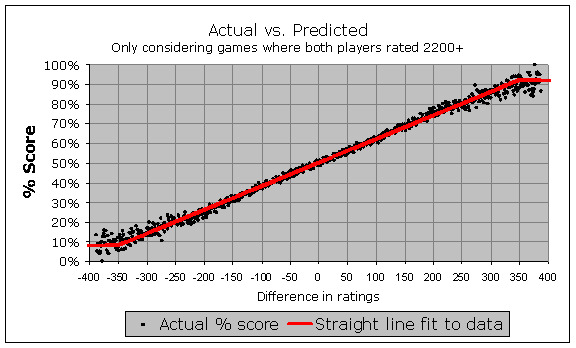}%

    \caption{Left: Elo skill curve (in yellow) overlayed with empirical data on the outcomes of FIDE-rated chess games.  Right: Sonas skill curve (in red) overlayed with empirical data on the outcomes of FIDE-rated chess games in which both players were rated 2200 or higher.  Figures by Jeff Sonas, taken from \cite{Sonas11}.  See \cite{Sonas02, Sonas20} for additional discussion.}
    \label{fig:my_label}
\end{figure}

In 2002, Jeff Sonas published a famous critique of Elo \cite{Sonas02}, in which he analyzed a large sample of FIDE rated chess games among highly-rated players.
He argued that the logistic curve is so flat within this player pool that a threshold-linear skill curve fits the data just as well, and hence is preferable due to its simplicity.
We will henceforth call skill curves that are linear within some threshold \emph{Sonas-like skill curves}:
\begin{definition} [Sonas-like skill curves] \label{def:sonaslike}
A skill curve is \textbf{Sonas-like} if there exist $a, s > 0$ such that
for ratings $x,y \in \rr$, when $|x-y| \le s$, we have $\sigma(x,y) = ax-ay + 0.5$.
\end{definition}

Sonas' rating system specifically took $s=400$, and assumes a flat skill curve outside this range; that is, $\sigma(x, y) = \sigma(y+400, y)$ when $x > y+400$.
In 2011 \cite{Sonas11}, Sonas augmented his previous analysis with an interesting nuance: when \emph{all} FIDE rated chess games are considered, including those where the participants have lower ratings, then the logistic curve is not so flat and Elo skill curves seem superior.
Thus, the relative value of Elo and Sonas seems to depend a bit on the strength of the pool of players being modeled. 

\subsection{Strategyproof Rating Systems and Our Results}

A system is strategyproof against a certain kind of undesirable strategic behavior if it does not incentivize a rating-maximizing player to perform that behavior.
We next overview a recent case study, in which players in a real-world system discovered and implemented an attack on a rating system that was missing an important strategyproofness property.

\subsubsection{Volatility Hacking}

An obviously-desirable strategyproofness property is that a rating-maximizing player should always be incentivized to try their best to win each game they play.
Interestingly, these incentives were recently discovered \emph{not} to hold in the popular Glicko-2 rating system.

Glicko-2 \cite{Glickman95} is a higher-dimensional version of Elo that tracks a \emph{volatility parameter} for each player.
A higher volatility parameter indicates that the system is more uncertain about the player's true rating, which in turn amplifies the adjustment function.
This volatility parameter was repeatedly shown to improve predictive accuracy of systems, and so it has been widely adopted.
But recently, an attack has been discovered called \emph{volatility hacking}.
A player first strategically loses many games, lowering their rating to the point where they only face much lower-rated opponents and can essentially win or lose at will.
Then, the player alternates wins and losses, which boosts their volatility parameter as much as they like while keeping their rating roughly stable.
Finally, they win a few games, and due to their enormous volatility parameter their rating skyrockets well above its initial value.

Volatility hacking was used to attack the game Pok{\'e}mon Go \cite{PokemonReddit}, which uses a version of Glicko-2 to rank players.
An interesting recent paper by Ebtekar and Liu \cite{EL21} shows that some other popular rating systems are also vulnerable to volatility hacking attacks, and it proposes concrete fixes to the handling of volatility parameters that would make a Glicko-2-like system (with a volatility parameter) strategyproof against volatility hacking.

\subsubsection{Opponent Indifference and First Main Result}

This paper is about securing rating systems against a less pernicious but likely much more common type of attack that we will call \emph{opponent selection}.
Many matchmaking systems provide players with some control over their next opponent.
For example, some online games use an interface in which players view a public list of active challenges from opponents of various ratings, and they can accept any challenge they like.
Alternately, some games use random matchmaking, but then let players abort a game without penalty after glimpsing their opponent's rating.
If players have any control over their next opponent's rating, then it is clearly undesirable for a rating system to allow a player to boost their rating by exercising this control strategically.
For example, in the Elo rating system, one can compute that the magnitude of expected gain is greatest when a (correctly-rated) opponent's true rating is exactly halfway between a player's current rating and their true rating.

Which rating systems are \emph{immune} to opponent selection incentives?
The natural definition would capture systems where the a player's expected gain doesn't depend on their opponent's rating, so long as that opponent is correctly rated.
Formally:
\begin{definition} [Opponent Indifference]
A rating system $(\sigma, \alpha)$ with expected gain function $\gamma$ is \textbf{opponent indifferent} if there is a function $\gamma^* : \rr^2 \to \rr$ such that
$$\gamma(x, x^* \mid y, y) = \gamma^*(x, x^*) \text{ for all } x, x^*, y.$$
\end{definition}

Before we proceed, let us be more specific about the kind of opponent selection attacks that become impossible in a rating system that satisfies this definition.
The premise is that the players in a system are not usually correctly rated, i.e., their current rating differs significantly from their true rating.
One reason is due to random walk effects, where a player's current rating will naturally drift over time.
Another reason is that a player's true rating might change from day to day due to external factors like rest, stress, distraction, etc.\footnote{Experts have estimated that chess strength fluctuates by about $\pm$ 200 Elo rating points (in the FIDE system) on a given day due to external factors \cite{COchess}.  However, we have not been able to find statistical validation of these estimates, so they should be taken as informal.}
An opponent selection attack would have a player seek opponents of one rating when they feel overrated by the system, and another when they feel underrated, causing their rating fluctuations to hold a significantly higher average than they would without strategic behavior.

Having motivated the axiom of opponent indifference, our next goal is to check whether it is satisfied by any actual rating systems, and whether these systems are ``reasonable.''
Although opponent indifferent rating systems exist (this follows from Lemma \ref{lemma:exist-op} in this paper), we actually argue that they are \emph{not} sufficiently expressive to be interesting.
Our reasoning is as follows.
Most interesting games exhibit \emph{skill chains}.
In chess, for example, a hobbyist can beat a beginner (say) 95\% of the time, and an expert can beat a hobbyist 95\% of the time, and a master can beat an expert 95\% of the time, and so on.
Most rating systems allow, in theory, for an infinite chain of players of ascending ratings, with each player beating the next with reasonably high probability.
We can formalize a version of this property as follows:
\begin{definition} [Full Scale] \label{def:fullscale}
A rating system $(\sigma, \alpha)$ has \textbf{full scale} if there exists $p > 0.5$ and an infinite ascending chain of ratings $r_1 < r_2 < \dots$ such that $\sigma\left(r_i, r_{i-1}\right) \ge p$ for all $i$.
\end{definition}

Our critique of opponent indifference is that it is incompatible with full scale, and thus cannot really capture games like chess.
\begin{theorem} [First Main Result] \label{thm:introOIFS}
There is no rating system that simultaneously satisfies opponent indifference and full scale.
\end{theorem}

Although we note that every major rating system we are aware of satisfies full scale, one could object at this point that full scale is too strong.
One might not expect an \emph{infinite} chain of ratings $r_1 < r_2 < \dots$ as in the definition of full scale to exist, but rather only a ``reasonably long'' chain of ratings.
The argument in our impossibility theorem (Theorem \ref{thm:nooifs}) arguably restricts these rating chains to be unreasonably short.
For example: if we take parameter $p=0.75$, i.e., we want a chain of ratings $r_1 < r_2 < \dots$ where each rating beats the previous one with probability $\ge 0.75$, this chain can have only length $3$.
This implies that an opponent indifferent rating system following the ELO convention where $\sigma(x+400, x) = 0.75$ would have a maximum possible rating spread of only $600$ points.

\subsubsection{$P$ Opponent Indifference and Second Main Result}

In light of Theorem \ref{thm:introOIFS}, our goal is now to weaken the definition of opponent indifference in a way that escapes the impossibility result of Theorem \ref{thm:introOIFS}, while still providing an effective type of immunity against opponent selection attacks.
Fortunately, we argue that a reasonable relaxation of opponent indifferent is already implicit in the \emph{thresholding effects} used in implementations of Elo and in Sonas.
These systems hardcode a rating difference at which the rating system punts: the system declares that chaotic effects dominate, and ratings are no longer a good predictor of outcome probabilities.
A natural relaxation is to enforce opponent indifference \emph{only} between opponents in the ``non-chaotic'' rating regime.

We call this weaker version \emph{$P$ opponent indifference}, where the parameter $P$ controls the threshold beyond which we no longer require opponent indifference to hold.
In the following, let us say that ratings $x, y$ are $P$-close if we have $\sigma(x, y) \in (0.5-P, 0.5+P)$.
Formally:\footnote{See Definition \ref{def:poi} for the equivalent technical definition used in the paper.}
\begin{definition} [$P$ Opponent Indifference] \label{pop_def}
For a parameter $P \in (0,0.5]$, a rating system is \textbf{$P$ opponent indifferent} if there is a function $\gamma^* : \rr^2 \to \rr$ such that
$$\gamma(x, x^* \mid y, y) = \gamma^*(x, x^*)$$
for any $x, x^*, y$ with $x,y$ and $x^*, y$ both $P$-close.
\end{definition}

$P$ opponent indifference might feel like a light tweak on opponent indifference;
it places the exact same requirements as opponent indifference on all ``reasonable'' games that might be played (according to a thresholding rule).
Thus it provides immunity against opponent selection attacks in the cases of interest.
But, perhaps surprisingly, this relaxation is enough to escape impossibility.
We prove the following characterization theorem.

\begin{definition} [$P$ Separable]
A skill curve $\sigma$ is \textbf{$P$ separable} if there is a weakly increasing function $\beta : \rr \to \rr$ such that $\sigma(x, y) = \beta(x) - \beta(y) + 0.5$ for all $P$-close $x, y$.
\end{definition}

\begin{definition} [$P$ Constant]\label{def:pconst}
For $P \in (0,0.5]$, a $K$-function is \textbf{$P$ constant} if for all $x,y$ such that $\sigma(x,y)\in(0.5-P,0.5+P)$, $K(x,y) = C$ for some constant C.
\end{definition}

\begin{theorem} [Second Main Result -- See Theorem \ref{theorem:characterize_pop} in the body]
A nontrivial rating system $(\sigma, \alpha)$ is $P$ opponent indifferent if and only if $\sigma$ is $P$ separable and its $K$-function is $P$ constant.
\end{theorem}


It is relatively easy from this characterization theorem to show that $P$ opponent rating systems exist, and that they can exhibit full scale (even with respect to any given parameter $P$).\footnote{This follows e.g.\ as a consequence of our Corollary \ref{cor:sonaschar}, which implies that the Sonas rating system itself is $P$ opponent indifferent (and it has full scale); we discuss this more next.}
In the next part, we show that these desirable properties continue to hold even under a natural \emph{strengthening} of $P$ opponent indifference.

\subsubsection{Strong $P$ Opponent Indifference and Third Main Result}

$P$ opponent indifference requires indifference between any two opponents within threshold, \emph{assuming those opponents are correctly rated}.
One might want to more strictly require indifference between two opponents that are both \emph{incorrectly rated by the same amount}.
Formally:
\begin{definition} [Strong $P$ Opponent Indifference]\label{intro_spoi}
For a parameter $P$, a rating system is \textbf{strongly $P$ opponent indifferent} if there is a function $\gamma^* : \rr^3 \to \rr$ such that
$$\gamma(x, x^* \mid y, y+\delta) = \gamma^*(x, x^*,\delta)$$
for all $x, x^*, y,y+\delta$ with $x, y$ and $x^*, y+\delta$ both $P$-close.
\end{definition}

Note that $P$ opponent indifference is the special case where $\delta=0$.
Fortunately, this strengthening remains tenable, as shown by the following theorem:

\begin{theorem} [Third Main Result -- See Corollary \ref{cor:sonaschar}] For any $0 < P \le 0.5$ and any rating system $(\sigma, \alpha)$, the rating system is strongly $P$ opponent indifferent if and only if it is Sonas-like with a $P$ constant $K$-function.
\end{theorem}

Conceptually, we interpret this theorem as a significant technical point in favor of Sonas.
For games where Elo and Sonas are equally preferable in terms of accurately modeling the player pool, the Sonas model has the additional advantage of fighting opponent selection attacks, which may give a reason why it should be favored.

In Corollary \ref{cor:soiimpossible}, we also discuss the analogous strengthening of (general) opponent indifference, and we strengthen our impossibility result to show that no nontrivial strong (general) opponent indifferent rating systems exist, whether or not they satisfy full scale.

\subsection{Further Details on Previous Work \label{sec:priorwork}}

Here we supply a bit more technical detail on some prior work referenced in the above discussion.

\paragraph{Applications.}

The most direct and common use of rating systems is in \emph{matchmaking}, where a system tracks estimates of player strength in order to create well-balanced and engaging matches.
Some recent work has focused on \emph{team} matchmaking, updating estimates of player strength when the system must create two opposing \emph{groups} worth of players.
See \cite{AM17} and references within for examples.

Rating systems have also been applied in economic contexts, such as pricing algorithms \cite{YDM14}.
Agents in the economy are given ratings, and items they might want to purchase are also given ratings.
An item is offered to an agent, at a price determined by the rating of the item.
If the agent buys the item, then the item ``wins'' the interaction; its rating (i.e., the system's appraisal of its value) increases, while the agent's rating (i.e., the system's estimate of their stinginess) decreases.
Alternately, if the agent refuses the item, then the agent ``wins'' the interaction, and the item's rating decreases.
A related use case is item response theory \cite{Forisek09}, where the goal is to categorize the difficulty of tasks, e.g., answering various questions on a test.
Here we again view each agent answering each question as a game, where the agent ``wins'' the interaction if they get the question right, and the question ``wins'' the interaction if the agent gets it wrong.
The system then updates its beliefs about the agent's capabilities or the question's difficulty accordingly.

\paragraph{Related Strategyproof Rating Mechanisms.}

Closely related to \emph{rating systems} are \emph{ranking systems}, where the goal is simply to order agents by relative skill (construed broadly; e.g., this would include rankings of colleges by ``quality'').
Although ranking systems often work by assigning agents a numeric score that could be interpreted as a rating, the main difference from our work is that they are \emph{offline}: they post-process data into a ranking, rather than updating ratings in an \emph{online} fashion in response to newly arriving data.
Agents who are aware of a ranking algorithm may act strategically to optimize their position, and thus, incentive-compatibility is an important feature of these algorithms.
For some work in this area, see \cite{LGB22} and references therein.

In our work, we do not assume any prior knowledge of agent skill levels, and we do not necessarily assume that agents even know their own skill level.
If we assume that system or agents \emph{do} have some prior knowledge of ratings, we might want a mechanism that incentivizes them to report their beliefs truthfully.
There is a related line of work developing such mechanisms \cite{EGL22, Su21, Su22}.

\subsection{Open Problems and Future Directions}

We conclude our introduction by suggesting a few open problems arising from our work.
\begin{enumerate}
\item What systems can satisfy opponent indifference properties if we relax our model to allow rating systems that are not one-dimensional, memoryless, or zero-sum?

\item Can opponent selection attacks be avoided by systems that separately track a \emph{display rating}, which is visible to players and which they may try to maximize, from a \emph{matchmaking rating}, used internally by the system to adjust ratings?
How different would the display ratings and matchmaking ratings need to be?\footnote{We are grateful to an anonymous reviewer who suggested this open problem.}

\item Can opponent indifference properties be achieved simultaneously with other desirable properties of rating systems?
Other natural properties to consider might include fast convergence of current ratings to true ratings, or other notions of strategyproofness, like the one considered by Ebtekar and Liu \cite{EL21}.
\end{enumerate}

\section{Preliminaries} 

Some of our results in the introduction reference \emph{nontrivial} rating systems.
We define these formally:
\begin{definition}[Trivial and nontrivial]
A rating system is \textbf{trivial} if for all $x,y$, $\sigma(x,y) = 0.5$, or \textbf{nontrivial} otherwise.
\end{definition} 

In the introduction, we discuss versions of properties that only hold over a restricted domain of ratings (like $P$ opponent indifference, $P$ constant, etc).
In the technical part of this paper, we will need to similarly generalize other functions, and so the following (slightly informal) language will be helpful:
\begin{definition} [Property over $(A,B)$]
Given a property defined over a set of inputs in $\rr$, we say that a property holds \textbf{over $(A,B)$} if it holds whenever the inputs are taken in the interval $(A,B)$.
We will equivalently say that $(A, B)$ is a property interval for the rating system.

For example, a rating system is nontrivial over $(A, B)$ if there exist $x,y \in (A,B)$ with $\sigma(x,y) \ne 0.5$, and in this case we say that $(A, B)$ is a nontrivial interval for the rating system.
\end{definition}

The definition of the expected gain function in the introduction is phrased in a way that makes its intuitive meaning clear.
However, the following alternate characterization in terms of the $K$-function will be useful in some of our proofs.

\begin{lemma}\label{exp-gain}
The expected gain function for a rating system $(\sigma, \alpha)$ satisfies
\begin{align*}
    \gamma(x,x^* \mid y,y^*) = K(x,y)(\sigma(x^*,y^*)-\sigma(x,y))
\end{align*}
\end{lemma}
\begin{proof}
We compute:
\begin{align*}
\gamma(x,x^* \mid y,y^*) &= \alpha(x,y)\sigma(x^*,y^*)-\alpha(y,x)\sigma(y^*,x^*)\\
&= K(x,y)\sigma(y,x)\sigma(x^*,y^*)-K(y,x)\sigma(x,y)\sigma(y^*,x^*)\tag*{Def of $\alpha$ using $K$-function}\\
&= K(x,y)(\sigma(y,x)\sigma(x^*,y^*)-\sigma(x,y)\sigma(y^*,x^*))\tag*{$K$ is symmetric}\\
&= K(x,y)((1-\sigma(x,y))\sigma(x^*,y^*)-\sigma(x,y)(1-\sigma(x^*,y^*)))\tag*{$\sigma$ Draw-Free}\\
&= K(x,y)(\sigma(x^*,y^*)-\sigma(x,y)\sigma(x^*,y^*)-\sigma(x,y)+\sigma(x,y)\sigma(x^*,y^*))\\
&= K(x,y)(\sigma(x^*,y^*)-\sigma(x,y)). \tag*{\qedhere}
\end{align*}
\end{proof}
\section{Opponent Indifference}
In this section we will characterize all opponent indifferent rating systems.
Some of our lemmas are proved with extra generality ``over $(A,B)$"; this generality will become useful in the following section.

\subsection{Opponent Indifference vs.\ Full Scale\label{op-vs-fs}}

This section is about the incompatibility between the opponent indifference and full scale axioms discussed in the introduction.
Our main technical lemma in this section is Lemma \ref{opntk-func}, which shows that any nontrivial opponent indifferent rating system has a constant $K$-function; this is proved by analyzing the expected gain function under the restriction implied by opponent indifference.
We then use this result about the $K$-function to reveal the structure of the skill curve. This ultimately makes full-scale impossible (Theorem \ref{thm:nooifs}).
\begin{lemma}\label{all-half}
If a rating system is both nontrivial and opponent indifferent over $(A,B)$ where $\sigma(x,z) = 0.5$ and $\sigma(y,z) = 0.5$, then $\sigma(x,y) = 0.5$ for $x,y,z \in (A,B)$.
\end{lemma}
\begin{proof}
Opponent indifference means that the expected gain only depends on the first two inputs, so
\begin{align*}
\gamma(x,y\mid z,z) &= \gamma(x,y\mid x,x)\tag*{Opponent Indifference}\\
K(x,z)(\sigma(y,z)-\sigma(x,z)) &= K(x,x)(\sigma(y,x)-\sigma(x,x)) \tag*{Lemma \ref{exp-gain}}\\
K(z,x)(0.5-0.5) &= K(x,x)(\sigma(y,x)-0.5)\tag*{$\sigma$ Draw-Free}\\
0 &= K(x,x)(\sigma(y,x)-0.5)
\end{align*}
Since the $K$-function is always positive as rating systems are consequential, we must have $\sigma(x,y) = 0.5$.
\end{proof}
\begin{lemma}\label{op-trivial-base}
If a rating system is both nontrivial and opponent indifferent over $(A,B)$, then for all $x \in (A,B)$, there exists $y \in (A,B)$ such that $\sigma(x,y) \ne 0.5.$
\end{lemma}
\begin{proof}
Suppose there was a rating $x \in (A,B)$ such that $\sigma(x,y) = 0.5.$ for all $y \in (A,B)$. There exists $z,z^* \in (A,B)$ such that $\sigma(z,z^*) \ne 0.5$ by nontriviality. We then have $\sigma(z,x) = 0.5$ and $\sigma(z^*,x) = 0.5$. This implies $\sigma(z^*,z)= 0.5$ by Lemma \ref{all-half} which is a contradiction.
\end{proof}
\begin{lemma}\label{op-trivial-cond-constant}
If a rating system is both nontrivial and opponent indifferent over $(A,B)$, then for $x,y \in (A,B)$ such that $\sigma(x,y) \ne 0.5$, $K(x,x) = K(x,y)$.
\end{lemma}
\begin{proof}
Let $x,y \in (A,B)$ where $\sigma(x,y) \ne 0.5$.
From the definition of opponent indifference, the expected gain does not depend on its latter two parameters.
We therefore have
\begin{align*}
\gamma(x, y \mid x, x) &= \gamma(x, y \mid y, y) \tag*{Opponent Indifference}\\
K(x,x)(\sigma(y,x)-\sigma(x,x)) &= K(x,y)(\sigma(y,y)-\sigma(x,y)) \tag*{Lemma \ref{exp-gain}}\\
K(x,x)(\sigma(y,x)-0.5) &= K(x,y)(0.5-\sigma(x,y))\\
K(x,x)(1-\sigma(x,y)-0.5) &= K(x,y)(0.5-\sigma(x,y))\tag*{$\sigma$ Draw-Free}\\
K(x,x)(0.5-\sigma(x,y)) &= K(x,y)(0.5-\sigma(x,y))\\
K(x,x) &= K(x,y). \tag*{\qedhere}
\end{align*}
\end{proof}
\begin{lemma}\label{opntk-func}
If a rating system $(\sigma, \alpha)$ is opponent indifferent over $(A,B)$ and nontrivial over $(A,B)$, then $K(x,y)$ is constant over $(A,B)$.
\end{lemma}
\begin{proof}
Let $x,y \in (A,B)$. We want to show that $K(x,x) = K(x,y)$. This means that the $K$-function depends only on its first input.
Since additionally $K$ is symmetric in its two parameters, it must in fact be constant over $(A,B)$. Lemma \ref{op-trivial-cond-constant} proves the statement in the case that $\sigma(x,y) \ne 0.5$, so we need only cover the case that  $\sigma(x,y) = 0.5$.

Lemma \ref{op-trivial-base} shows that there exists $z\in (A,B)$ such that $\sigma(x,z) \ne 0.5$. Lemma \ref{all-half} implies $\sigma(z,y) \ne 0.5$ by contrapositive as $\sigma(x,z) \ne 0.5$ and $\sigma(x,y) = 0.5$.
We can then evaluate
\begin{align*}
\gamma(x,z \mid y,y) &= \gamma(x,z \mid z,z)\\
K(x,y)(\sigma(z,y)-\sigma(x,y)) &= K(x,z)(\sigma(z,z)-\sigma(x,z)) \\
K(x,y)(\sigma(z,y)-0.5) &= K(x,z)(0.5-\sigma(x,z))
\end{align*}
and
\begin{align*}
\gamma(y,z \mid z,z) &= \gamma(y,z \mid x,x)\\
K(y,z)(\sigma(z,z)-\sigma(y,z)) &= K(y,x)(\sigma(z,x)-\sigma(y,x)) \\
K(y,z)(0.5 -\sigma(y,z)) &= K(y,x)(\sigma(z,x)-0.5)\\
K(y,z)(\sigma(z,y)- 0.5) &= K(y,x)(0.5 - \sigma(x,z))
\end{align*}
Since the $K$-function is symmetric, we let $a = K(x,y) = K(y,x)$ and by Lemma \ref{op-trivial-cond-constant}, we let $b = K(x,z) = K(y,z)$. We then let $c=(\sigma(z,y)- 0.5)$ and $d= (0.5 - \sigma(x,z))$, giving us
$$a \cdot c = b \cdot d$$
and 
$$b \cdot c = a \cdot d$$
The previous two equations imply that $a=b$, and thus $K(x,y) = K(x,z) = K(x,x)$ by Lemma \ref{op-trivial-cond-constant}.
\end{proof}
We comment that the nontriviality hypothesis in the previous lemmas is indeed necessary, since otherwise $K(x,y)$ can be essentially any symmetric function over $(A,B)$ and the expected gain will still be $0$.

\begin{lemma}\label{lemma:relation-xyz}
If a rating system $(\sigma, \alpha)$ is opponent indifferent over $(A,B)$, then
$$\sigma(x,y) = \sigma(x,z)-\sigma(y,z) + 0.5$$
for all $x,y,z \in (A,B)$.
\end{lemma}
\begin{proof}
If the system is trivial over $(A,B)$ then the lemma holds immediately, so assume nontriviality.
From the definition of opponent indifference the expected gain only depends on a player's current rating and true rating, $x$ and $x^*$, so long as all ratings are within the interval $(A,B)$.
So we have:
\begin{align*}
\gamma(x, x^* \mid y, y) &= \gamma(x, x^* \mid z, z) \tag*{Opponent Indifference}\\
    K(x,y) (\sigma(x^*,y)-\sigma(x,y)) &= K(x,z)  (\sigma(x^*,z)-\sigma(x,z)) \tag*{Lemma \ref{exp-gain}}\\
    C \cdot (\sigma(x^*,y)-\sigma(x,y)) &= C \cdot (\sigma(x^*,z)-\sigma(x,z)) \tag*{Lemma \ref{opntk-func}.}
\end{align*}
Now consider the possible setting $x^* = y$.
Under this, we continue:
\begin{align*}
C \cdot (\sigma(y,y)-\sigma(x,y)) &= C \cdot (\sigma(y,z)-\sigma(x,z))\\
0.5-\sigma(x,y) &= \sigma(y,z)-\sigma(x,z)\\
\sigma(x,y) &= \sigma(x,z)-\sigma(y,z) + 0.5. \tag*{\qedhere}
\end{align*}
\end{proof}
\begin{definition}[Separable]
A skill curve is \textbf{separable} if there is a function $\beta: \rr \to [C,C+0.5]$ for some constant C such that 
$$\sigma(x, y) = \beta(x) - \beta(y) + 0.5.$$
for all $x,y$. $\beta$ is called a \textbf{bisector} of the skill curve.
\end{definition}

Note that bisectors are not unique: if $\beta$ is a bisector for $\sigma$, then any vertical translation of $\beta$ is also a bisector.
However, all bisectors differ by this translation:

\begin{lemma} \label{lem:bisectorshift}
For any skill curve $\sigma$ that is separable over $(A,B)$, if $\beta, \beta'$ are both bisectors of $\sigma$ over $(A,B)$, then there is a constant $C$ such that $\beta'(x) = \beta(x) + C$ for all $x \in (A,B)$.
\end{lemma}
\begin{proof}
Fix an arbitrary $y \in (A, B)$.
We have
\begin{align*}
\sigma(x, y) = \beta'(x) - \beta'(y) + 0.5 = \beta(x)-\beta(y) + 0.5
\end{align*}
and so, rearranging, we get
\begin{align*}
\beta'(x) = \beta(x) + (\beta'(y) - \beta(y)).
\end{align*}
Now the claim follows by taking $C := \beta'(y) - \beta(y)$.
\end{proof}

\begin{lemma}\label{op-sep}
If a rating system $(\sigma,\alpha)$ is opponent indifferent over $(A,B)$, then $\sigma$ is separable over $(A,B)$ and all bisectors are continuous. 
\end{lemma}
\begin{proof}
Suppose $(\sigma,\alpha)$ is opponent indifferent over $(A,B)$. Let $m \in (A,B)$ be an arbitrary constant. From Lemma \ref{lemma:relation-xyz} we have
$$\sigma(x,y) = \sigma(x,m) - \sigma(y,m) + 0.5.$$
Therefore $\beta(x) := \sigma(x,m)$ is a bisector for $\sigma$ over $(A,B)$, and since $\sigma$ is continuous, $\beta$ is continuous as well.
Finally, Lemma \ref{lem:bisectorshift} implies that since one bisector is continuous, all bisectors are continuous.
\end{proof}

\begin{lemma}\label{lemma:exist-op}
If a rating system's $K$-function is constant over $(A,B)$ and the skill curve is separable over $(A,B)$ then the rating system is opponent indifferent over $(A,B)$.
\end{lemma}
\begin{proof}
Plugging the equations into the expected gain function for $x,x^*,y \in (A,B)$, we have
\begin{align*}
    \gamma(x, x^* \mid y, y) &= K(x,y) (\sigma(x^*,y)-\sigma(x,y)) \tag*{Lemma \ref{exp-gain}}\\
    &= C \cdot (\beta(x^*)-\beta(y)+0.5-(\beta(x)-\beta(y)+0.5)) \\
    &= C \cdot (\beta(x^*)-\beta(x)).
\end{align*}
Thus $\gamma$ depends only on its first two parameters $x$ and $x^*$, implying opponent indifference.
\end{proof}

Note that Lemma \ref{lemma:exist-op} implies that nontrivial opponent indifferent rating systems do indeed exist, as we can choose a constant $K$-function and suitable bisector function $\beta$ over the entire interval $(-\infty, \infty)$.
However, as discussed in the introduction, our next theorem shows that no opponent indifferent rating system can exhibit full scale:

\begin{theorem} \label{thm:nooifs}
No opponent indifferent rating system (over $(-\infty, \infty)$) has full scale.
\end{theorem}
\begin{proof}
Let $(\sigma,\alpha)$ be an opponent indifferent rating system and fix some $p \in (0.5,1)$.
Our plan is to prove that for some integer $N$ depending only on $p$, there does not exist a chain of ratings $r_1 < \dots < r_N$ such that for all $1 <i \le N$, we have $\sigma(r_i, r_{i-1}) = p$.
More specifically, our strategy is to prove that
\begin{equation}
\sigma(r_N,r_1) \ge (N-1)p - \frac{N-2}{2}. \label{eq:ind}
\end{equation}
Since $\sigma(r_N, r_1) \le 1$, this implies
\begin{align*}
(N-1)p - \frac{N-2}{2} \le 1\\
Np - p - \frac{N}{2} + 1 \le 1\\
N\left(p - \frac{1}{2}\right) \le p\\
N \le \frac{2p}{2p-1}
\end{align*}
which is an upper bound for $N$ depending only on $p$, as desired.
It now remains to prove equation (\ref{eq:ind}).
We do so by induction on $N$:
\begin{itemize}
\item (Base Case, $N = 1$) We have $\sigma(r_1,r_1) = 1/2$, as desired.
\item (Inductive Step)
By Lemma \ref{op-sep} the skill curve is separable, and so for some bisector $\beta$ we have
\begin{align*}
\sigma(r_{N+1},r_1) 
&= \sigma(r_{N+1},r_N) - \sigma(r_1,r_N) + 0.5 \tag*{Lemma \ref{lemma:relation-xyz}}\\
&\ge p - (1-\sigma(r_N,r_1)) + 0.5 \tag*{$\sigma$ Draw-Free}\\
&\ge p - \left(1-\left((N-1)p - \frac{N-2}{2}\right)\right) + 0.5 \tag*{Inductive Hypothesis}\\
&= p - 1 + (N-1)p - \frac{N-2}{2} + 0.5\\
&= Np -\frac{N-1}{2}
\end{align*}
which verifies (\ref{eq:ind}) for $N+1$. \qedhere
\end{itemize}
\end{proof}

We note the setting $N \le \frac{2p}{2p-1}$ that arises in this theorem, which for a given $p$ controls the maximum possible length of a skill chain that may exist in an opponent indifferent system.
In particular, plugging in $p=0.75$, we get $N \le \frac{1.5}{0.5} = 3$.

\subsection{Impossibility of Strong Opponent Indifference}

Here, we consider an even stronger notion of opponent indifference.
While opponent indifference enforces that we are indifferent between two \emph{correctly-rated} opponents, we could more strongly require indifference between two opponents \emph{misrated by the same amount}.
In particular, the following definition enforces indifference between opponents who are overrated by $\delta$:

\begin{definition} [Strong Opponent Indifference]
A rating system is \textbf{strongly opponent indifferent} if  there is a function $\gamma^* : \rr^3 \to \rr$ such that, for all $x, x^*, y, y+\delta$, we have
$$\gamma(x,x^*\mid y,y+\delta) = \gamma^*(x, x^*,\delta).$$
\end{definition}

A strongly opponent indifferent rating system is also opponent indifferent (by considering $\delta=0$), and thus by Theorem \ref{thm:nooifs}, it cannot exhibit full scale.
We will prove an even stronger impossibility theorem.
While this is a slight detour (as we have arguably already shown that strong opponent indifference is undesirable), the results in this section are used in the following section.
As before, we prove our intermediate lemmas over arbitrary intervals $(A, B)$, which provides flexibility that will be helpful later in the paper.

\begin{definition} [Translation Invariant] \label{def:translationinvariant}
A rating system $(\sigma,\alpha)$ is \textbf{translation invariant} if there is a function $\sigma^* : \rr \to [0, 1]$ with $\sigma(x, y) = \sigma^*(x-y)$ for all $x, y$.
\end{definition}
\begin{theorem}\label{sop-t}
If a rating system is strongly opponent indifferent over $(A,B)$, then it is translation invariant over $(A,B)$.
\end{theorem}
\begin{proof}
If the system is trivial over $(A,B)$ then the lemma is immediate, so assume nontriviality over $(A,B)$.
Let $(\sigma,\alpha)$ be a strongly opponent indifferent rating system.
To show translation invariance over $(A,B)$, let $x, x^*, \delta$ be such that $x,x^*,x+\delta,x^*+\delta \in (A,B)$.
Since strong opponent indifference implies opponent indifference, Lemmas \ref{exp-gain} and \ref{opntk-func} apply.
We can therefore compute:
\begin{align*}
    \gamma(x, x^* \mid y, y+\delta) &= K(x,y) (\sigma(x^*,y+\delta)-\sigma(x,y)) \tag*{Lemma \ref{exp-gain}}\\
    &= C \cdot (\sigma(x^*,y+\delta)-\sigma(x,y)) \tag*{Lemma \ref{opntk-func}.}
\end{align*}

From the definition of strong opponent indifference, the expected gain function depends only on the first two parameters and the difference between its latter two parameters, and so we have
\begin{align*}
\gamma(x, x^* \mid x^*, x^*+\delta) &= \gamma(x, x^* \mid x, x+\delta)\\
C \cdot (\sigma(x^*,x^*+\delta)-\sigma(x,x^*)) &= C \cdot (\sigma(x^*,x+\delta )-\sigma(x,x)) \tag*{previous equation}\\
\sigma(x^*,x^* + \delta)-\sigma(x,x^*)&= \sigma(x^*,x + \delta )-0.5\\
\sigma(x,x^*) &= \sigma(x,x^* + \delta)- \sigma(x^*,x + \delta ) + 0.5\\
\sigma(x,x + \delta)- \sigma(x^*,x + \delta ) + 0.5 &= \sigma(x^*,x^* + \delta)- \sigma(x^*,x + \delta ) + 0.5 \tag*{Lemma \ref{lemma:relation-xyz}}\\
\sigma(x,x + \delta)&=  \sigma(x^*,x^* + \delta),
\end{align*}
which implies translation invariance over $(A, B)$.
\end{proof}


\begin{theorem}\label{sop-linear}
If a rating system $(\sigma, \alpha)$ is strongly opponent indifferent over $(A,B)$, then $\sigma$ has a bisector that is linear over $(A,B)$ (i.e., on the interval $(A, B)$ it coincides with a function of the form $\beta(x) = mx$ for some constant $m$).
\end{theorem}
\begin{proof}
By considering an appropriate horizontal translation of our rating system, we may assume without loss of generality that our opponent indifferent interval is symmetric about the origin, i.e., it has the form $(-A, A)$.
Let $\beta$ be an arbitrary function that bisects $\sigma$ over $(-A, A)$, and further assume without loss of generality that $\beta$ intersects the origin, i.e., $\beta(0)=0$.
Our goal is to prove that, for all $x, y \in (-A, A)$, we have $\beta(x+y) = \beta(x) + \beta(y)$.
Since we already have that $\beta$ is continuous, this implies linearity of $\beta$ over $(-A, A)$ (see Lemma \ref{app:linearize} in Appendix \ref{A} for details).

Let $x, y, x+y \in (-A, A)$.
We then have
\begin{align}
    \sigma(x+y,x) = f(y) \tag*{Theorem \ref{sop-t}}\\
    \beta(x+y)-\beta(x) = f(y)\tag*{Lemma \ref{op-sep}}\\
    \beta(x+y) = \beta(x)+f(y) \label{eq:xdelta}.
\end{align}
Additionally, by reversing the roles of $x$ and $y$, we have
\begin{align}
    \beta(x+y)= f(x)+\beta(y) \label{eq:deltax}.
\end{align}
Thus, combining (\ref{eq:xdelta}) and (\ref{eq:deltax}), we get
\begin{align*}
    f(x) + \beta(y) &= \beta(x) + f(y).
\end{align*}
It follows from equation (\ref{eq:xdelta}) that $f(0) = 0$, and we recall that $\beta(0) = 0$.
Therefore, plugging $y=0$ into the previous equation, we get
\begin{align*}
    f(x) + \beta(0) &= \beta(x) + f(0)\\
    f(x)  &= \beta(x)
\end{align*}
and so, recombining with (\ref{eq:xdelta}), we have $\beta(x+y) = \beta(x)+\beta(y)$, as desired.
\end{proof}

\begin{corollary} \label{cor:soiimpossible}
A strongly opponent indifferent rating system (over $(-\infty, \infty)$) must be trivial.
\end{corollary}
\begin{proof}
If a rating system is strongly opponent indifferent, then by Theorem \ref{sop-linear} there exists a bisector $\beta$ that is linear over $(-\infty, \infty)$.
Additionally, the range of $\beta$ is contained within $[C,C+0.5]$ for some constant $C$.
The only such functions are those with slope $0$.
Thus
$$\sigma(x,y) = 0 - 0 + 0.5 = 0.5$$
for all $x,y \in \rr$, and so it is trivial.
\end{proof}

\section{$P$ Opponent Indifference}


Recall that $P$ opponent indifference is a relaxation of (full) opponent indifference, in which we only require indifference among opponents who are reasonably evenly matched.

\begin{definition} [$P$ Opponent Indifference]\label{def:poi}
For $P \in (0,0.5]$, a rating system is \textbf{$P$ opponent indifferent} if the rating system is opponent indifferent over $(A,B)$ for all $A<B$ where\footnote{Technically, this definition differs slightly from Definition \ref{pop_def} of partial opponent indifference in the introduction, because it forces all pairs among $x, x^*, y$ to be $P$-close, whereas Definition \ref{pop_def} only explicitly forces $x, y$, and $x^*, y$ to be $P$-close.  However, it follows as an easy corollary of the results in this section that the two definitions are equivalent (see Lemma \ref{app:poi} in Appendix \ref{A} for details).}
$$\sigma(A, B) > 0.5-P.$$
\end{definition}

\subsection{Characterization of $P$ Opponent Indifference}

Our next goal is to prove the following characterization theorem:
\begin{theorem}[Characterization of $P$ Opponent Indifferent Rating Systems]\label{theorem:characterize_pop}
A nontrivial rating system is $P$ opponent indifferent if and only if the skill curve is $P$ separable and the $K$-function is $P$ constant.
\end{theorem}

We work towards a proof with some intermediate structural lemmas.

\begin{lemma}\label{lemma:pop_reals}
If a rating system $(\sigma,\alpha)$ is $P$ opponent indifferent then every real number is within an opponent indifferent interval.
\end{lemma}
\begin{proof}
Let $b\in \rr$. Since $\sigma$ is continuous, there exists $0<\epsilon$ such that $0.5\ge\sigma(b-\epsilon, b+\epsilon)> 0.5-P$. Thus $(b-\epsilon, b+\epsilon)$ is an opponent indifferent interval.
\end{proof}

\begin{lemma}\label{p_sep}
If a rating system $(\sigma,\alpha)$ is $P$ opponent indifferent, then $\sigma$ is $P$ separable with a continuous bisector.
\end{lemma}

Let us quickly discuss this lemma statement before we begin its proof.
By Lemma \ref{op-sep}, for each $P$-interval $(A, B)$, there exists a continuous function $\beta$ that bisects $\sigma$ on this interval.
But this lemma is claiming something stronger: that there is a \emph{single} continuous function $\beta$ that bisects $\sigma$ on \emph{all} $P$-intervals simultaneously.
This strengthening requires a bit of topology to prove, but is overall fairly straightforward.

\begin{proof} [Proof of Lemma \ref{p_sep}]
We construct a bisector $\beta$ as follows.
First, we arbitrarily fix a point; say, $\beta(0)=0$, and all other values of $\beta$ are currently undefined.
Then, iterate the following process.
Let $x$ be the current supremum of the points on which $\beta$ has been defined.
Let $\eps>0$ be the largest value such that $(x-\eps, x+\eps)$ is a $P$-interval.
Using Lemma \ref{op-sep}, there exists an infinite family of bisectors over $(x-\eps, x+\eps)$, which are vertical translations of each other.
However, since $(x-\eps, x+\eps)$ intersects at least one previously-defined point, only one bisector in this family is consistent with the previously-selected values of $\beta$.
We may extend $\beta$ to the interval $(x-\eps, x+\eps)$ using this particular bisector, and repeat infinitely.

We claim that, for every nonnegative number $r$, there is a constant $c_r$ such that the value of $\beta(r)$ is defined after finitely many iterations of this process.
To see this, let $D \subseteq \rr$ be the set of points with this property, and suppose for contradiction that the supremum of $D$ is a finite real number $r^*$.
Note that $D$ is the union of open intervals, and thus $D$ is open, which means it does not contain its supremum.
Since $r^* \notin D$, there is no constant $c_{r^*}$ for which $\beta(r^*)$ is defined after $c_{r^*}$ iterations.
Additionally, since $r^*$ is the supremum of $D$, any interval of the form $(r^* - \eps, r^* + \eps)$ intersects $D$.
Let $r \in (r^* - \eps, r^* + \eps) \cap D$.
Then after $c_r+1$ iterations, we would define $\beta(r^*)$.
Since $c_r+1$ is finite, this implies $r^* \in D$, which completes the contradiction.

By a symmetric process, we may then extend $\beta$ to negative inputs.
\end{proof}

\begin{lemma}\label{lemma:constant-pk}
For a nontrivial $P$ opponent indifferent rating system, every opponent indifferent interval $(A,B)$ has a constant $K$-function over $(A,B)$.
\end{lemma}
\begin{proof}
The rating system is either nontrivial or trivial over $(A,B)$. If the rating system is nontrivial, the result is immediate by Lemma \ref{opntk-func}.
Thus the remaining case is when the system is trivial over $(A, B)$, but nontrivial overall.

We have $\sigma(A, B) = 0.5$.
Since $\sigma$ is not identically $0.5$, there exist $B' > B, A' < A$ such that $\sigma(A', B') < 0.5$.
Moreover, since $\sigma$ is continuous, we may specifically choose values $A', B'$ satisfying
$$0.5-P < \sigma(A', B') < 0.5.$$
So $(A', B')$ is a $P$-interval, and $\sigma$ is nontrivial over $(A', B')$.
By Lemma \ref{opntk-func}, we thus have a constant $K$-function over $(A', B')$.
Since $(A, B) \subseteq (A', B')$, the lemma follows.
%
%
\end{proof}
We are now ready to prove Theorem \ref{theorem:characterize_pop}.
\begin{proof} [Proof of Theorem \ref{theorem:characterize_pop}]
($\longrightarrow$) Lemma \ref{p_sep} implies that $\sigma$ is $P$ separable, so it only remains to show that the $K$-function is $P$ constant.
We prove this in two steps.
First, let $f(x) = K(x, x)$, and we will prove that $f(x)$ is constant.
Since every $x$ lies in an opponent indifferent interval (Lemma \ref{lemma:pop_reals}), by Lemma \ref{lemma:constant-pk} $f$ is constant on that interval.
In particular this implies $f$ is differentiable at $x$, with $f'(x)=0$, and thus $f$ is constant.
Finally: for any $x, y$ in the same opponent indifferent interval $(A, B)$, by Lemma \ref{lemma:constant-pk} we have $K(x, y) = K(x, x)$, and thus the $K$-function is $P$ constant.

($\longleftarrow$) This direction is immediate from Lemma \ref{lemma:exist-op}.
\end{proof}

\subsection{Characterization of Strong $P$ Opponent Indifference}
In this section, we will show how the strong version of P opponent indifference characterizes Sonas-like curves.
We begin with an auxiliary lemma:

\begin{lemma}\label{lemma:sufficient-sop}
If a rating system's $K$-function is constant over $(A,B)$ and the skill curve is separable over $(A,B)$ with a linear bisector, then the rating system is strongly opponent indifferent over $(A,B)$.
\end{lemma}
\begin{proof}
Plugging the equations into the expected gain function for $x,x^*,y,y+\delta \in (A,B)$, we have
\begin{align*}
    \gamma(x, x^* \mid y, y+\delta) &= K(x,y) (\sigma(x^*,y+\delta)-\sigma(x,y)) \tag*{Lemma \ref{exp-gain}}\\
    &= C \cdot (\beta(x^*)-\beta(y+\delta)+0.5-(\beta(x)-\beta(y)+0.5)) \\
    &= C \cdot (\beta(x^*)-\beta(x) - \beta(\delta)) 
\end{align*}
which depends only on $x, x^*$, and $\delta$.
\end{proof}

\begin{definition}[Strong $P$ Opponent Indifference]\label{def:spoi}
For $P \in (0,0.5]$, a rating system is \textbf{strongly $P$ opponent indifferent} if the rating system is strongly opponent indifferent over $(A,B)$ for all A and B where\footnote{Similar to $P$ opponent indifference, this definition is equivalent to the one given in the introduction, despite looking slightly different (see Lemma \ref{app:spoi} in Appendix \ref{A} for details).}
$$\sigma(A, B) > 0.5-P.$$
\end{definition}

\begin{theorem}[Characterization of Strongly $P$ Opponent Indifferent Rating Systems]\label{spop-char}
A nontrivial continuous rating system $(\sigma, \alpha)$ is strongly $P$ opponent indifferent if and only if the skill curve is $P$ separable with a linear bisector and the $K$-function is $P$ constant.
\end{theorem}
\begin{proof} ($\longrightarrow$)
Assume the rating system is strongly $P$ opponent indifferent. Because strongly $P$ opponent indifferent rating systems are also $P$ opponent indifferent, we already have that the skill curve is $P$ separable and the $K$-function is $P$ constant by Theorem \ref{theorem:characterize_pop}.
It only remains to prove is that the skill curve has a linear bisector.

The skill curve has a continuous bisector $\beta$ by Lemma \ref{p_sep}.
Recall from Theorem \ref{sop-linear} that, for any strongly opponent indifferent interval $(A, B)$, $\sigma$ is separable on this interval with a bisector that is linear over the interval.
Moreover, the slope of $\beta$ on this interval must be $\frac{\sigma(B,A) - 0.5}{B-A}$.

From Lemma \ref{lemma:pop_reals} and the fact that the rating system is nontrivial, we know there exists some $0<\epsilon$ and $a \in \rr$ such that $(a-\epsilon,a+\epsilon)$ is a nontrivial strongly opponent indifferent interval.
Since the bisector is linear and nontrivial over $(a-\epsilon,a+\epsilon)$, we know $ 0.5-P< \sigma(a-\epsilon,a+\epsilon) < 0.5$, and thus the slope of $\beta$ is positive over $(a-\epsilon,a+\epsilon)$.
Because of this, there must be some value $c$ where $\sigma(a,c) = \sigma(a-\epsilon,a+\epsilon)$. Thus $(a-\epsilon,a+\epsilon)$ and $(a,c)$ are overlapping strongly opponent indifferent intervals so $\beta$ must have the same slope over both intervals. In order for this to be the case, $c = a+2\eps$.

We can continue this argument, repeatedly adding $\eps$, to show that $\beta$ has the same slope as the interval $(a-\eps,a+\eps)$ for all $r>a$. We can also make a similar argument subtracting $\eps$ each time to show that the bisector has the same slope as the interval $(a-\eps,a+\eps)$ for $r<a$. Thus, the bisector's slope is constant. Since any vertical translation of $\beta$ is also a bisector, there exists a linear bisector.

($\longleftarrow$) This direction follows directly from Lemma $\ref{lemma:sufficient-sop}$.
\end{proof}

\begin{definition}[$P$ Translation Invariant]
A rating system $(\sigma,\alpha)$ is \textbf{$P$ translation invariant} if for every $x,y \in \rr$ where $\sigma(x,y) \in (0.5-P,0.5+P)$, we have $\sigma(x,y) = \sigma^*(x-y)$ for some $\sigma^*: \rr \to [0,1]$.
\end{definition}

\begin{lemma}\label{sonas-like}
A skill curve is Sonas-like with parameters $a, s$ if and only if it is $\sigma(s,0) - 0.5$ separable with $\beta(x) = ax$.
(Sonas-like skill curves are defined in Definition \ref{def:sonaslike}.)
\end{lemma}
\begin{proof}
($\longrightarrow$) This direction follows from Definition \ref{def:sonaslike}, as $\sigma(s+x,x)$ is constant for all x.\\
($\longleftarrow$) From the definition of $P$ separable and linear bisectors, for all $x,x+\delta$ where $\sigma(x,x+\delta) \in (0.5-P,0.5+P),$ we have
$$\sigma(x,x+\delta) = \beta(x)-\beta(x+\delta) + 0.5 = -\beta(\delta)+0.5,$$
and so the rating system is $P$ translation invariant. Thus, for all $|x-y| < s$ we have $\sigma(x,y) = ax-ay + 0.5$.
Since $\sigma$ is continuous, it additionally holds that $\sigma(x,y) = ax-ay + 0.5$ when $|x-y| = s$.
\end{proof}

\begin{corollary} \label{cor:sonaschar}
A rating system is strongly $P$ opponent indifferent if and only if it is Sonas-like with a $P$ constant $K$-function.
\end{corollary}
\begin{proof}
This follows directly from Theorem $\ref{spop-char}$ and Lemma $\ref{sonas-like}$.
\end{proof}
\section*{Acknowledgments}

We are grateful to Jeff Sonas for a helpful discussion on the history and relationship between rating systems, and we are grateful to an anonymous reviewer for thorough and helpful feedback.

\bibliographystyle{plainurl}
\bibliography{references}

\appendix

\section{Extra Proofs} \label{A}

\begin{lemma} \label{app:linearize}
 If a function $f:(-A,A) \to \rr$ is continuous and for $x,y \in (-A,A)$ we have $f(x+y) = f(x) +f(y)$, then $f$ is linear over $(-A,A)$.
\end{lemma}
\begin{proof}
 To prove linearity, we need only prove $f(rx) = r \cdot f(x)$ for all $x, rx \in (-A, A)$. To do this we will prove the following propositions in order:
\begin{enumerate}
    \item $f(nx) = n \cdot f(x)$ for $n,x$ such that $n \in \mathbb{Z}$ and $x,nx \in (-A,A)$
    \item $f(qx) = q \cdot f(x)$ for $q,x$ such that $q \in \mathbb{Q}$ and $x,qx \in (-A,A)$
    \item $f(rx) = r \cdot f(x)$ for $r,x$ such that $n \in \rr$ and $x,rx \in (-A,A)$.
\end{enumerate}

The first proposition follows directly from the fact that $f(x+y)= f(x)+f(y)$, and it can be used to prove the second proposition. Let $q= \frac{a}{b}$ such that $a,b \in \mathbb{Z}$ .
\begin{align*}
    b \cdot f(qx) &= b \cdot f(\frac{a}{b} \cdot x)\\
    &= ab \cdot f(\frac{1}{b} \cdot x) \tag*{$f(\frac{1}{b} \cdot x)$ is defined}\\
    &= a\cdot f(b \cdot \frac{1}{b} \cdot x)\\
    &= a \cdot f(x).
\end{align*}
Dividing both sides of the equation by b gives
\begin{align*}
    f(qx) &= \frac{a}{b} \cdot f(x)\\
    &= q \cdot f(x). \\
\end{align*}

To show proposition 3, we can use proposition 2 and continuity. Since $\mathbb{Q}$ is dense in $\rr$, we know there exists a sequence of rationals $q_n$ that converges to any $r \in \rr$. Thus we can say
\begin{align*}
    f(rx) &= f\left( \lim_{n\to\infty} q_n x\right)\\
    &= \lim_{n\to\infty} f \left(q_n x\right) \tag*{continuity of $f$}\\
    &= \lim_{n\to\infty} q_n \cdot f\left(x\right) \tag*{proposition 2}\\
    &= r \cdot f \left(x\right).
\end{align*}

So $f$ must be linear over $(-A,A)$.
\end{proof}
\begin{lemma} \label{app:poi}
Definition \ref{pop_def} and Definition \ref{def:poi} are equivalent.
\end{lemma}
\begin{proof}
 ($\longrightarrow$) Definition \ref{pop_def} is immediately stronger than Definition \ref{def:poi} since the former requires $P$-closeness among fewer numbers.
 
 ($\longleftarrow$) If a rating system satisfies Definition \ref{def:poi}, then we may apply Theorem \ref{theorem:characterize_pop}.
 Thus $\sigma(x,y) = \beta(x)-\beta(y) + 0.5$ and $K(x,y) = C$ for some constant $C$ for $x,y$ $P$-close. Looking at $x,y$ and $x^*,y$ both $P$-close for $x,x^*,y$ as defined in Definition \ref{pop_def}, we evaluate
 \begin{align*}
 \gamma(x,x^* \mid y, y) &= K(x,y)(\sigma(x^*,y)-\sigma(x,y)) \tag*{Lemma \ref{exp-gain}}\\
 &= C \cdot (\beta(x^*) - \beta(y) + 0.5 -(\beta(x)-\beta(y) +0.5)) \tag*{Theorem \ref{theorem:characterize_pop}}\\
  &= C \cdot (\beta(x^*) -\beta(x))
 \end{align*}
 which only depends on the first two parameters, therefore satisfying Definition \ref{pop_def}.
\end{proof}
\begin{lemma}\label{app:spoi}
Definition \ref{intro_spoi} and Definition \ref{def:spoi} are equivalent.
\end{lemma}
\begin{proof}
 ($\longrightarrow$) Definition \ref{intro_spoi} is immediately stronger than Definition \ref{def:spoi}, since the former requires $P$-closeness among fewer numbers.
 
 ($\longleftarrow$) If a rating system satisfies \ref{def:spoi}, then we may apply Theorem \ref{spop-char}.
 Thus $\sigma(x,y) = \beta(x)-\beta(y) + 0.5$ where $\beta$ is linear and $K(x,y) = C$ for some constant $C$ for $x,y$ $P$-close. Looking at $x,y$ and $x^*,y+\delta$ both $P$-close for $x,x^*,y,\delta$ as defined in Definition \ref{pop_def}, we evaluate
 \begin{align*}
 \gamma(x,x^* \mid y, y+\delta) &= K(x,y)(\sigma(x^*,y+\delta)-\sigma(x,y)) \tag*{Lemma \ref{exp-gain}}\\
 &= C \cdot (\beta(x^*) - \beta(y+\delta) + 0.5 -(\beta(x)-\beta(y) +0.5)) \tag*{Theorem \ref{spop-char}}\\
  &= C \cdot (\beta(x^*) - (\beta(y) + \beta(\delta)) + 0.5 -(\beta(x)-\beta(y) +0.5)) \tag*{$\beta$ linear}\\
  &= C \cdot (\beta(x^*) -\beta(x) - \beta(\delta))
 \end{align*}
 which only depends on the first two parameters and the difference of the latter two parameters, therefore satisfying Definition \ref{pop_def}.
\end{proof}
\end{document}